\documentclass[10pt]{article}
\usepackage{amsmath,amssymb,amsthm,stmaryrd,graphicx,verbatim,charter,xifthen,multirow,complexity,mathtools,fullpage}

\newtheorem{problem}{Problem}
\newtheorem{lemma}{Lemma}
\newtheorem{theorem}{Theorem}
\newtheorem{observation}{Observation}

\newtheorem{corollary}{Corollary}

\newcommand{\ksum}[1]{\ifthenelse{\isempty{#1}}{$k$-{\sf SUM}}{$k$-{\sf SUM(#1)}}}
\newcommand{\kldt}[1]{\ifthenelse{\isempty{#1}}{$k$-{\sf LDT}}{$k$-{\sf LDT(#1)}}}
\newcommand{\gkldt}[1]{\ifthenelse{\isempty{#1}}{{\sf GENERALIZED} $k$-{\sf LDT}}{{\sf GENERALIZED} $k$-{\sf LDT(#1)}}}

\newcommand{\xsum}[2]{\ifthenelse{\isempty{#2}}{#1-{\sf SUM}}{#1-{\sf SUM(#2)}}}
\newcommand{\xldt}[2]{\ifthenelse{\isempty{#2}}{#1-{\sf LDT}}{#1-{\sf LDT(#2)}}}

\newcommand{\threesum}[1]{\xsum{$3$}{#1}}
\newcommand{\threeldt}[1]{\xldt{$3$}{#1}}

\title{Geometric Pattern Matching Reduces to $k$-{\sf SUM}}
\author{Boris Aronov\thanks{Department of Computer Science and Engineering, Tandon School of Engineering, New York
University, Brooklyn, NY 11201, USA; {\tt boris.aronov@nyu.edu}. Partially supported by NSF grant CCF-15-40656 and by grant 2014/170
from the US-Israel Binational Science Foundation. Work by B.A. on this paper has been partially carried out while visiting ULB in November-December 2019, with support from ULB and F.R.S.-FNRS (Fonds National de la Recherche Scientifique).} 
  \and Jean Cardinal\thanks{Universit\'e libre de Bruxelles (ULB), Brussels, Belgium; {\tt jcardin@ulb.ac.be}. Supported by the F.R.S.-FNRS (Fonds National de la Recherche Scientifique) under CDR Grant J.0146.18.}}
\begin{document}
\maketitle
\sloppy

\begin{abstract}
We prove that some exact geometric pattern matching problems reduce in linear time to \ksum{} when the pattern has a fixed size $k$. 
This holds in the real RAM model for searching for a similar copy of a set of $k\geq 3$ points within a set of $n$ points in the plane, and for searching for an affine image of a set of $k\geq d+2$ points within a set of $n$ points in $d$-space.\\

As corollaries, we obtain improved real RAM algorithms and decision trees for the two problems. In particular, they can be solved by algebraic decision trees of near-linear height.
\end{abstract}

\section{Introduction}

The \ksum{} problem is a fixed-parameter version of the \NP-complete {\sf SUBSET SUM} problem.
It consists of deciding, given a set of $n$ numbers, whether any subset of size $k$ sum to zero. The problem for $k=3$, known as \threesum{}, is now
a well-established bottleneck problem in fine-grained complexity theory (see for instance~\cite{AWY18,W18} and references therein).
While there are many reductions showing \threesum{}- or \ksum{}-hardness of computational problems in geometry, only few reductions {\em to}
\threesum{} and \ksum{} are known. We give examples of computational geometry problems that reduce to \threesum{} or \ksum{}.

Our results are motivated by the nontrivial improved upper bounds on the complexity of \threesum{} and \ksum{} proven in the recent years.
While it has long been conjectured that no subquadratic algorithm for \threesum{} existed, it is now known to be solvable in time $O((n^2/\log n)(\log \log n)^2)$ in the real RAM model~\cite{GP18,F17,GS17,C18}. The existence of an $O(n^{2-\delta})$ algorithm for some $\delta>0$ remains an open problem.
Using folklore {\em meet-in-the-middle} algorithms, \ksum{} can be solved in time $O(n^{\lceil k/2\rceil})$ if $k$ is odd, and in time $O(n^{k/2}\log n)$ if $k$ is even.
Recently, Kane, Lovett, and Moran~\cite{KLM19} showed that it can be solved in time $O(n\log^2 n)$ in the linear decision tree model, improving on previous polynomial bounds~\cite{CIO16,ES19}.

\paragraph{Geometric pattern matching.}

We consider two problems involving searching for a given set $P$ of $k$ points, called the {\em pattern}, within a larger set $S$ of points, up to some geometric transformation.
Here we focus on {\em exact} algorithms, in which the pattern must match the subset of points exactly.
We consider the following two problems.

\begin{problem}[{\sf SIMILARITY MATCHING}]
  For a fixed integer $k\geq 3$, given a set $P$ of $k$ points in the plane and a set $S$ of $n$ points in the plane,
  determine whether $S$ contains the image of $P$ under a similarity transformation.
\end{problem}

\begin{problem}[{\sf AFFINE MATCHING}]
For fixed integers $d\geq 2$ and $k\geq d+2$, given a set $P$ of $k$ points in $\mathbb{R}^d$ containing $d+1$ affinely independent points, and
a set $S$ of $n$ points in $\mathbb{R}^d$, determine whether $S$ contains the image of $P$ under an affine transformation.
\end{problem}

A large body of the computational geometry and pattern recognition literature is dedicated to the problems of finding {\em approximate}
matches up to some geometric transformation, where the quality of the approximation is typically measured by the Hausdorff distance~\cite{CGHKKK97,GMO99,GIMV03,AK09}.
For exact pattern matching problems under different families of transformations, known upper bounds on time complexity have been compiled in a survey by Peter Bra\ss~\cite{B02}.
We reproduce them in Table~\ref{tab:ub}.

\begin{table}
  \begin{center}
  \begin{tabular}{|c|c|c|}
    \hline
    Transformations & Dimension & Complexity \\
    \hline
    congruence & 2 & $O(kn^{4/3}\log n)$~\cite{B02} \\
    congruence & 3 & $O(kn^{5/3}\log n 2^{O(\alpha(n)^2)})$~\cite{AS02} \\
    translation & $d$ & $O(kn\log n)$ (easy) \\
    homothety & $d$ & $O(kn^{1+1/d}\log n)$~\cite{EE94,B02} \\
    similarity & $d$ & $O(kn^d\log n)$~\cite{B02}\\
    affine & $d$ & $O(kn^{d+1}\log n)$~\cite{B02} \\
    \hline
  \end{tabular}
  \end{center}
  \caption{\label{tab:ub}Known upper bounds on the time complexity of exact geometric pattern matching in various settings (taken from~\cite{B02} and~\cite{HB17}, Chapter 54). We indicate the dependency on the pattern size $k$.}
  \end{table}

The complexity of these algorithms are directly related to bounds on the maximum number of occurrences of a pattern or a distance in a set of $n$ points.
In fact, such bounds directly yield a lower bound on the computational problem of {\em listing} all occurrences of the pattern.
A prototypal example is Erd\H{o}s' unit distance problem; see Bra\ss\ and Pach~\cite{BP05} for more examples.
It is known, in particular, that there can be $\Theta(n^2)$ similar copies of a pattern in an $n$-point set~\cite{EE94,AF00,AFKK16}.
Structural results on the extremal point sets are also known~\cite{AEF04}.
For affine transformations in $\mathbb{R}^d$, there exist pairs $P,S$ such that $S$ contains $\Theta (n^{d+1})$ copies of $P$: for instance the
$d$-dimensional lattice $\{1,2,\ldots ,n^{1/d}\}^d$ contains $\Theta (n^{d+1})$ affine images of a cube. 

\paragraph{Our results.}

We suppose we can perform exact computations over the reals.
Therefore, all the algorithms that we consider are either uniform algorithms in the real RAM model, or nonuniform algorithms in the algebraic decision tree model.

Our main result is the following.
\begin{theorem}
  \label{thm:main}
  {\sf SIMILARITY MATCHING} and {\sf AFFINE MATCHING} reduce in randomized linear time to \ksum{}.
  \end{theorem}

We refer the reader to the exact definitions of the \ksum{} problem and the notion of randomized linear-time reduction given later.
Theorem~\ref{thm:main} has a number of consequences.
For instance, combining the reduction provided by Theorem~\ref{thm:main} with the real RAM algorithm for \threesum{} from Chan~\cite{C18}, we obtain the following.

\begin{corollary}
  There exists an $O((n^2/\log n)(\log \log n)^2)$ randomized real RAM algorithm for {\sf SIMILARITY MATCHING} with a pattern size $k=3$.
  In particular, there exists a subquadratic algorithm to detect equilateral triangles in a point set.
\end{corollary}

This contrasts with our current knowledge on the related \threesum{}-hard problem of finding three collinear points, also known as {\sf GENERAL POSITION TESTING}.
Despite recent attempts~\cite{BCILOS19,C18}, it is still an open problem to find a subquadratic algorithm for {\sf GENERAL POSITION TESTING}.

Our next corollary is obtained directly from known algorithms for \ksum{}.
It improves on the best known $O(n^{d+1}\log n)$ algorithm whenever $k<2(d+1)$.
\begin{corollary}
  There exists an $O(n^{\lceil k/2\rceil})$ (for $k$ odd), or an $O(n^{k/2}\log n)$ (for $k$ even) randomized real RAM algorithm for {\sf AFFINE MATCHING}.
\end{corollary}

Finally, we consider the nonuniform {\em decision tree} complexity, also known as {\em query complexity}, of the two problems.
By applying a recent result of Kane, Lovett, and Moran~\cite{KLM19}, we can bound the number of algebraic tests that are required to detect copies of $P$ in an input set $S$.

\begin{corollary}
\label{cor:tree}
There exist algebraic decision trees of height $O(n\log^2 n)$ for {\sf SIMILARITY MATCHING} and {\sf AFFINE MATCHING}.
\end{corollary}

In fact, if the pattern $P$ is a fixed parameter, that is, when $P$ is not part of the input, but known at the algorithm design time, then the decision tree in the statement above only involves {\em linear} tests.

\begin{corollary}
\label{cor:treeFP}
There exist linear decision trees of height $O(n\log^2 n)$ for the fixed-parameter versions of {\sf SIMILARITY MATCHING} and {\sf AFFINE MATCHING}, in which $P$ is a fixed parameter of the problems.
\end{corollary}

In the case $k=3$, {\sf SIMILARITY MATCHING} is one of the two similarity testing problems recently tackled by Aronov, Ezra, and Sharir~\cite{AES20}.
They consider the problem of deciding, given three sets $A,B,C$ of $n$ points in the plane, whether there exists $(a,b,c)\in A\times B\times C$ that simultaneously satisfies two real polynomial equations.
They provide a subquadratic upper bound on the algebraic decision tree complexity of this problem.
We observe that {\sf SIMILARITY MATCHING} with $k=3$ can be cast as such a problem in which the two equations are linear, and that in that case the decision tree complexity becomes near-linear. We therefore improve on one of the results of Corollary 4.4 in~\cite{AES20}.
A thorough discussion of the relation between the two problems and that of testing polynomials for vanishing on product point sets can be found in the full version of the paper~\cite{AES20}.

\paragraph{Plan.} In the next section, we define a number of variants of the \ksum{} problem and prove they are all equivalent in the computation model we consider. In Section~\ref{sec:similar}, we prove our main result for {\sc SIMILARITY MATCHING}. Section~\ref{sec:affine} considers the {\sc AFFINE MATCHING} problem. The last section is dedicated to the proof of Corollaries~\ref{cor:tree} and~\ref{cor:treeFP}.

\section{Linear degeneracy testing}
\label{sec:ldt}

We first give a definition of the \ksum{} problem.
Here, $k\geq 3$ is a fixed integer, and $X$ is a ring. 

\begin{problem}[\ksum{$X$}]
Given $k$ sets $A_1,\ldots ,A_k$ of $n$ elements of $X$, determine whether there exists a $k$-tuple $a_1,\ldots ,a_k\in\vartimes_{i=1}^k A_i$ such that $\sum_{i=1}^k a_i =0$.
\end{problem}

Our next problem is often referred to as {\em linear degeneracy testing}~\cite{AC05,DGS20}.
We consider the cases where $X=\mathbb{R}$ or $\mathbb{C}$ with the usual addition and multiplication operations, or where $X=\mathbb{R}^d$ or $\mathbb{C}^d$ for some integer $d\geq 2$, with the vector addition and Hadamard (entrywise) product defined by $(uv)_i = u_iv_i$.
In the latter cases, the all-zero vector is denoted by $0$, and the all-one vector by $1$.

\begin{problem}[\kldt{$X$}]
For a linear function $f\colon X^k\to X$ given by $f(a_1,\ldots ,a_k)= \beta_0 + \sum_{i=1}^k \beta_i a_i$ with $\beta_i\in X$ for $0\leq i\leq k$,
given $k$ sets $A_1,\ldots ,A_k$ of $n$ elements of $X$, determine whether there exists a $k$-tuple $a_1,\ldots ,a_k\in\vartimes_{i=1}^k A_i$ such that $f(a_1,\ldots ,a_k)=0$.
\end{problem}

We make two observations. First, these are {\em fixed-parameter problems}: the integer $k$ is part of the definition of the problem, not of the input. The same can be assumed for the function $f$.
Such parameters will be referred to as {\em fixed} in what follows.
Another observation is that using the Hadamard product in the definition of the function $f$ allows us to combine conditions on the sought $k$-tuples: In the ring $X$, searching for $k$-tuples that simultaneously satisfy $d$ linear equations can be cast as \kldt{$X^d$}.

It is clear that \ksum{} is the special case of \kldt{} in which $\beta_0=0$ and $\beta_i=1$ for $1\leq i\leq k$.
On the other hand, \kldt{} is not harder than \ksum{}.
\begin{lemma}
\label{lem:ldt2sum}
For any integer $d>0$, \kldt{$X$} reduces in linear time to \ksum{$X$}.
\end{lemma}
\begin{proof}
Consider the sets $A_i$ from the \kldt{} instance, and let $A'_i \coloneqq \{\beta_i a \mid a\in A_i \}$ for all $1\leq i< k$, and $A'_k\coloneqq \{\beta_k a+\beta_0\mid a\in A_k\}$.
Then the instance of \ksum{} composed of the sets $A'_i$ has a solution if and only if the instance of \kldt{} has a solution.
\end{proof}

In what follows, we say that a problem {\sf A} {\em reduces to problem {\sf B} in randomized $g(n)$ time}
if there exists an algorithm in the real RAM model with access to random real numbers in $[0,1]$ that 
maps any instance of size $n$ of {\sf A} to an equivalent instance of {\sf B} in time $O(g(n))$ with probability 1.
If we insist on using only random bits, we can make the error probability arbitrary small by using sufficiently many random bits.

Over the reals, the vector and scalar versions of \ksum{} are also essentially equivalent, up to such a randomized reduction.
\begin{lemma}
\label{lem:vec2sc}
For any fixed integer $d>0$, \ksum{$\mathbb{R}^d$} reduces in randomized linear time to \ksum{$\mathbb{R}$} .
\end{lemma}
\begin{proof}
Given an instance $\{A_1,\ldots ,A_k\}$ of \ksum{$\mathbb{R}^d$}, pick a uniform random unit vector $v\in\mathbb{R}^d$ and 
consider the sets $A'_i \coloneqq \{a\cdot v \mid a\in A_i\}\subset \mathbb{R}$ (here $a\cdot v$ is the usual dot product).
They form an instance of \ksum{$\mathbb{R}$} such that any solution to the original instance of \ksum{$\mathbb{R}^d$} is also a solution.
In the other direction, suppose there is a $k$-tuple $a'_1,\ldots ,a'_k\in\vartimes_{i=1}^k A'_i$ such that $\sum_{i=1}^k a'_i =0$, where $a'_i=a_i\cdot v$.
Hence we have $\sum_{i=1}^k a_i\cdot v=0$, which is either because $v\perp \sum_{i=1}^k a_i$ and $\sum_{i=1}^k a_i\not= 0$, or because $\sum_{i=1}^k a_i=0$.
Since $v\perp \sum_{i=1}^k a_i$ and $\sum_{i=1}^k a_i\not= 0$ occurs with probability 0, the $k$-tuple $a_1,\ldots ,a_k$ is a solution of the instance 
$\{A_1,\ldots ,A_k\}$ of \ksum{$\mathbb{R}^d$} with probability 1.
\end{proof}

We also make the following simple observation:

\begin{observation}
\label{obs:c2r}
\ksum{$\mathbb{C}^d$} is equivalent to \ksum{$\mathbb{R}^{2d}$}.
\end{observation}

\section{Searching for a similar copy}
\label{sec:similar}

We first consider the special case of the {\sc SIMILARITY MATCHING} problem in which $k=3$.
\begin{problem}[{\sf TRIANGLE}]
Given a triangle $\Delta$ and a set $S$ of $n$ points in the plane, determine whether $S$ contains three points whose convex hull is similar to $\Delta$.
\end{problem}

The short proof of the following result uses the interpretation of points in the plane as complex numbers, an idea that was exploited in a combinatorial context before~\cite{EE94,LR97}.
\begin{lemma}
{\sf TRIANGLE} reduces in linear time to \threesum{$\mathbb{C}$}.
\end{lemma}
\begin{proof}
Let $u=re^{i\theta}$ be such that the three numbers $0,1,u$ are the vertices of a triangle similar to $\Delta$ in the complex plane.
Recall that multiplying by $re^{i\theta}$ has a geometric interpretation in the complex plane as scaling by a factor $r$ and rotating by an angle $\theta$. 
Hence three other complex numbers $a,b,c\in\mathbb{C}$ form a triangle similar to $\Delta$ in the complex plane with the same orientation if and only if $c-a=u(b-a)$, or equivalently if $(u-1)a-u b+c = 0$.
Hence {\sf TRIANGLE} reduces to \threeldt{$\mathbb{C}$} with $\beta = (u-1, -u, 1)$.
From Lemma~\ref{lem:ldt2sum}, it reduces in linear time to \threesum{$\mathbb{C}$}.
\end{proof}

Combining with Observation~\ref{obs:c2r} and Lemma~\ref{lem:vec2sc}, we obtain:
\begin{theorem}
{\sf TRIANGLE} reduces in randomized linear time to \threesum{$\mathbb{R}$}.
\end{theorem}

This generalizes naturally to larger patterns.
\begin{lemma}
  \label{lem:sim2ksum}
{\sf SIMILARITY MATCHING} reduces in linear time to \ksum{$\mathbb{C}^{k-2}$}. 
\end{lemma}
\begin{proof}
Let $u_1,\ldots ,u_{k-2}\in\mathbb{C}$ be such that the set $Q=\{0,1,u_1,\ldots ,u_{k-2}\}$ is similar to $P$ in the complex plane.
Then $k$ numbers $a_1,\ldots ,a_k\in\mathbb{C}$ form a similar copy of $Q$ in the complex plane, with $a_1$ mapped to $0$, $a_2$ to $1$, and so on, if and only if $a_i - a_1 = u_{i-2} (a_2 - a_1)$ for all $3\leq i\leq k$. 
These are $k-2$ linear equations on the $k$ complex numbers $a_1,\ldots ,a_k$, hence {\sf SIMILARITY MATCHING} reduces in linear time to \kldt{$\mathbb{C}^{k-2}$}. 
From Lemma~\ref{lem:ldt2sum}, it reduces in linear time to \ksum{$\mathbb{C}^{k-2}$}.
\end{proof}

Again, combining with Observation~\ref{obs:c2r} and Lemma~\ref{lem:vec2sc}, we obtain the first statement of Theorem~\ref{thm:main}.
\begin{theorem}
  \label{thm:mainsim}
{\sf SIMILARITY MATCHING} reduces in randomized linear time to \ksum{$\mathbb{R}$}.
\end{theorem}

\section{Searching for an affine image}
\label{sec:affine}

We now prove the analogous result for the affine case.
As a warm-up, we first consider the following simpler special case of {\sf AFFINE MATCHING} in which the pattern is a square.
Four points form the affine image of vertices of a square if and only if they are the vertices of a (possibly degenerate) parallelogram.
Hence the problem can be cast as follows.

\begin{problem}[{\sf PARALLELOGRAM}]
Given a set $S$ of $n$ points in the plane, determine whether $S$ contains four points whose convex hull is a parallelogram.
\end{problem}

\begin{theorem}
{\sf PARALLELOGRAM} reduces in randomized linear time to \xsum{$4$}{$\mathbb{R}$}.
\end{theorem}
\begin{proof}
Four points $a_1,a_2,a_3,a_4\in S$ in this order form a parallelogram with $a_1a_2$ parallel to $a_4a_3$ and $a_2a_3$ parallel to $a_1a_4$ if and only if $a_2-a_1 = a_3-a_4$, or equivalently if $a_1-a_2+a_3-a_4=0$.
Hence {\sf PARALLELOGRAM} reduces to \xldt{4}{$\mathbb{R}^2$} with $\beta = ((0,0),(1,1),(-1,-1),(1,1),(-1,-1))$.
From Lemmas~\ref{lem:ldt2sum} and~\ref{lem:vec2sc}, it also reduces in randomized linear time to \xsum{$4$}{$\mathbb{R}$}.
\end{proof}

The general case follows from the following observation.
Consider a matrix $Q\in\mathbb{R}^{n\times n}$, and let $Q_k$ denote the matrix obtained from $Q$ by replacing its $k$th column by the column vector $x^T$, where $x_1, x_2,\ldots ,x_n$ are variables.
Then $\det Q_k$ is a linear combination of $x_1, x_2,\ldots ,x_n$, with coefficients defined by $Q$.

\begin{lemma}
  \label{lem:aff2ksum}
{\sf AFFINE MATCHING} reduces in linear time to \ksum{$\mathbb{R}^{\ell}$} with $\ell =d(k-(d+1))$.
\end{lemma}
\begin{proof}
We use the notation $[k]\coloneqq \{1,2,\ldots ,k\}$.
Let $p_i=(p_{i,1},\ldots ,p_{i,d})$ be a row vector representing the $i$th point of $P$.
From the problem definition, $P$ must contain $d+1$ affinely independent points.
We assume, without loss of generality, that these points are the first $d+1$ points $p_1,\ldots ,p_{d+1}$.
Let $A=\{a_1,\ldots ,a_k\}\in {S\choose k}$ be a candidate match.
In order for the set $A$ to be the image of $P$ under an affine transformation, there must be a solution to the system of $k$ linear equations of the form $p_iF + t=a_i$ for all $i\in [k]$, with $d^2+d$ real unknowns $F\in \mathbb{R}^{d\times d}$ and $t\in\mathbb{R}^d$.
The system can be decomposed into $d$ systems, one for each coordinate $j\in [d]$.
Each consists of $k$ equations with $d+1$ unknowns, of the form $p_i F_j + t_j = a_{ij}$ for $i\in [k]$, where $F_j$ is the $j$th column of $F$. 
We consider one such system, for a fixed $j\in [d]$, and restrict it to the first $d+1$ equations only:
$$
Q \cdot 
\begin{pmatrix}
F_j \\
t_j
\end{pmatrix}
=
\begin{pmatrix}
a_{1,j}\\
\vdots \\
a_{d+1,j}
\end{pmatrix},
\text{\ where\ }
Q=
\begin{pmatrix}
p_1 & 1 \\
\vdots & \vdots \\
p_{d+1} & 1
\end{pmatrix} .
$$

Since the first $d+1$ points of $P$ are affinely independent, $Q$ is invertible and the system defines a unique solution for the coefficients $F_j$ and $t_j$ of the affine transformation.
From Cramer's rule, the value of the $k$th unknown is the ratio $\det Q_k / \det Q$, where $Q_k$ is the matrix 
obtained by replacing the $k$th column of $Q$ by $(a_{1,j},\ldots ,a_{d+1,j})^T$.
From the above observation and the fact that $Q$ does not depend on $S$, the expressions $\det Q_k / \det Q$ are linear combinations of the values $a_{1,j},\ldots ,a_{d+1,j}$, with coefficients determined by $P$.
Hence the explicit solution for the coefficients $F_j$ and $t_j$ are linear combinations of the $a_{1,j},\ldots ,a_{d+1,j}$.

A necessary and sufficient condition for the set $A$ to be a match is that the remaining $k-d-1$ points of $A$ are also images of the corresponding points in $P$.
Hence we require that for all $i>d+1$ the $i$th equation
$p_i F_j + t_j = a_{ij}$ is also satisfied by this solution.
The unknowns $F_j$ and $t_j$ can be replaced by linear combinations of $a_{1,j},\ldots ,a_{d+1,j}$.
Hence we obtain a set of $k-(d+1)$ linear equations on the variables $a_{1,j},\ldots ,a_{k,j}$, with coefficients depending on $P$.

Since these $k-(d+1)$ equations must hold for all coordinates $j\in [d]$ simultaneously, we obtain that 
{\sf AFFINE MATCHING} reduces to \kldt{$\mathbb{R}^{\ell}$} with $\ell =d(k-(d+1))$.
From Lemma~\ref{lem:ldt2sum} it also reduces to \ksum{$\mathbb{R}^{\ell}$}.
Since $d$ and $k$ are fixed, the reduction takes linear time.
\end{proof}

Combining with the randomization step in Lemma~\ref{lem:vec2sc}, we get the second part of Theorem~\ref{thm:main}.

\begin{theorem}
  \label{thm:mainaff}
  {\sf AFFINE MATCHING} reduces in randomized linear time to \ksum{$\mathbb{R}$}.
\end{theorem}

\section{Algebraic decision tree complexity}

An {\em algebraic decision tree} is a type of nonuniform algorithm for problems on inputs composed of $n$ real numbers.
For each input size $n$, it consists of a binary tree whose internal nodes are labeled with inequalities of the form ``$q(x)\leq 0$'' on the input $x\in\mathbb{R}^n$, where $q$ is a bounded-degree $n$-variate polynomial in $x_1,x_2,\ldots ,x_n$.
Inequalities are interpreted as {\em queries} on the input, and the two subtrees correspond to the possible outcomes of the query on the input.
Leaves of the tree are labeled with the answer to the problem.
The minimum height $h(n)$ of an algebraic decision tree solving instances of size $n$ the problem is the {\em decision tree} complexity, or {\em query} complexity of the problem.
When the queries only involve linear functions, such trees are called {\em linear decision trees}.
In that case, a query is said to be {\em $t$-sparse} when it involves at most $t$ numbers of the input.

We have the following recent result on the linear decision tree complexity of the \ksum{} problem.

\begin{theorem}[Kane, Lovett, Moran~\cite{KLM19}]
  \label{thm:ksumldt}
The \ksum{} problem on $n$ elements can be solved by a linear decision tree of height $O(n\log^2 n)$ in which
all the queries are $2k$-sparse and have only $\{-1,0,1\}$ coefficients.
\end{theorem}

We now show that this result directly applies to the {\sf SIMILARITY MATCHING} and {\sf AFFINE MATCHING} problems, thereby proving Corollary~\ref{cor:tree}.

We first consider the {\sf SIMILARITY MATCHING} problem, an instance $y$ of which consists of two coordinates per point of $P$ and $S$, hence of $2(k+n)$ real numbers.
Suppose we apply the reduction proposed in Theorem~\ref{thm:mainsim} to obtain an instance of \ksum{$\mathbb{R}$}.
Now consider the linear decision tree from Theorem~\ref{thm:ksumldt}.
Each linear query on the transformed input maps to a query on the original input numbers $y$.
Because the reduction only involves multiplications and additions on these numbers, such queries are algebraic queries on the original input $y$.
Therefore, the linear decision tree for \ksum{} maps to an algebraic decision tree of the same height for {\sf SIMILARITY MATCHING}.
The same reasoning applies to {\sf AFFINE MATCHING}.
In that case, it suffices to observe that multiplying both sides of every query by the quantity $\det Q$ for the matrix $Q$ used in the proof of Lemma~\ref{lem:aff2ksum} yields algebraic queries again.
Note that since $k$ and $d$ are constant and the linear queries in Theorem~\ref{thm:ksumldt} are sparse, the queries have bounded degree and bounded size.
This proves Corollary~\ref{cor:tree}.

Also note that if we suppose the pattern $P$ is a fixed parameter of the problem, then the two problems are solved by {\em linear} decision trees of height $O(n\log^2 n)$. It can indeed be checked that the algebraic queries do not involve multiplications between coordinates of the points of $S$, hence are linear whenever $P$ is fixed. This proves Corollary~\ref{cor:treeFP}.
It applies in particular to the {\sf PARALLELOGRAM} problem, or for finding an equilateral triangle in a point set.

\bibliographystyle{plain}
\bibliography{GeomPatternMatching}

\begin{thebibliography}{10}

\bibitem{AWY18}
Amir Abboud, Virginia~Vassilevska Williams, and Huacheng Yu.
\newblock Matching triangles and basing hardness on an extremely popular
  conjecture.
\newblock {\em {SIAM} J. Comput.}, 47(3):1098--1122, 2018.

\bibitem{AEF04}
Bernardo~M. {\'{A}}brego, Gy{\"{o}}rgy Elekes, and Silvia
  Fern{\'{a}}ndez{-}Merchant.
\newblock Structural results for planar sets with many similar subsets.
\newblock {\em Combinatorica}, 24(4):541--554, 2004.

\bibitem{AF00}
Bernardo~M. {\'{A}}brego and Silvia Fern{\'{a}}ndez{-}Merchant.
\newblock On the maximum number of equilateral triangles,~{I}.
\newblock {\em Discrete {\&} Computational Geometry}, 23(1):129--135, 2000.

\bibitem{AFKK16}
Bernardo~M. {\'{A}}brego, Silvia Fern{\'{a}}ndez{-}Merchant, Daniel~J. Katz,
  and Levon Kolesnikov.
\newblock On the number of similar instances of a pattern in a finite set.
\newblock {\em Electr. J. Comb.}, 23(4):P4.39, 2016.

\bibitem{AS02}
Pankaj Agarwal and Micha Sharir.
\newblock The number of congruent simplices in a point set.
\newblock {\em Discrete \& Computational Geometry}, 28(2):123--150, 2002.

\bibitem{AK09}
Dror Aiger and Klara Kedem.
\newblock Geometric pattern matching for point sets in the plane under
  similarity transformations.
\newblock {\em Inf. Process. Lett.}, 109(16):935--940, 2009.

\bibitem{AC05}
Nir Ailon and Bernard Chazelle.
\newblock Lower bounds for linear degeneracy testing.
\newblock {\em J. {ACM}}, 52(2):157--171, 2005.

\bibitem{AES20}
Boris Aronov, Esther Ezra, and Micha Sharir.
\newblock Testing polynomials for vanishing on cartesian products of planar
  point sets, 2020.
\newblock ar{X}iv:2003.09533. An earlier version of this paper appeared in
  {S}o{CG}'{20}.

\bibitem{BCILOS19}
Luis Barba, Jean Cardinal, John Iacono, Stefan Langerman, Aur{\'{e}}lien Ooms,
  and Noam Solomon.
\newblock Subquadratic algorithms for algebraic {3SUM}.
\newblock {\em Discrete {\&} Computational Geometry}, 61(4):698--734, 2019.

\bibitem{B02}
Peter Bra{\ss}.
\newblock Combinatorial geometry problems in pattern recognition.
\newblock {\em Discrete {\&} Computational Geometry}, 28(4):495--510, 2002.

\bibitem{BP05}
Peter Bra\ss\ and J{\'a}nos Pach.
\newblock Problems and results on geometric patterns.
\newblock In David Avis, Alain Hertz, and Odile Marcotte, editors, {\em Graph
  Theory and Combinatorial Optimization}, pages 17--36. Springer US, Boston,
  MA, 2005.

\bibitem{CIO16}
Jean Cardinal, John Iacono, and Aur{\'{e}}lien Ooms.
\newblock Solving k-{SUM} using few linear queries.
\newblock In Piotr Sankowski and Christos~D. Zaroliagis, editors, {\em 24th
  Annual European Symposium on Algorithms, {ESA} 2016, August 22-24, 2016,
  Aarhus, Denmark}, volume~57 of {\em LIPIcs}, pages 25:1--25:17. Schloss
  Dagstuhl - Leibniz-Zentrum f{\"{u}}r Informatik, 2016.

\bibitem{C18}
Timothy~M. Chan.
\newblock More logarithmic-factor speedups for {3SUM}, (median, +)-convolution,
  and some geometric {3SUM}-hard problems.
\newblock In {\em Proceedings of the Twenty-Ninth Annual {ACM-SIAM} Symposium
  on Discrete Algorithms, {SODA} 2018, New Orleans, LA, USA, January 7-10,
  2018}, pages 881--897, 2018.

\bibitem{CGHKKK97}
L.~Paul Chew, Michael~T. Goodrich, Daniel~P. Huttenlocher, Klara Kedem, Jon~M.
  Kleinberg, and Dina Kravets.
\newblock Geometric pattern matching under euclidean motion.
\newblock {\em Comput. Geom.}, 7:113--124, 1997.

\bibitem{DGS20}
Bartłomiej Dudek, Paweł Gawrychowski, and Tatiana Starikovskaya.
\newblock All non-trivial variants of 3-{LDT} are equivalent.
\newblock In {\em Proceedings of the 52nd Annual {ACM} Symposium on Theory of
  Computing {STOC}, June 22–26, 2020, Chicago ({IL})}, 2020.

\bibitem{EE94}
Gy\"{o}rgy Elekes and Paul Erd\H{o}s.
\newblock {\em Similar Configurations and Pseudogrids}, pages 85--104.
\newblock Colloquia Mathematica Societatis J\'{a}nos Bolyai, North Holland,
  Amsterdam, 1994.

\bibitem{ES19}
Esther Ezra and Micha Sharir.
\newblock A nearly quadratic bound for point-location in hyperplane
  arrangements, in the linear decision tree model.
\newblock {\em Discrete {\&} Computational Geometry}, 61(4):735--755, 2019.

\bibitem{F17}
Ari Freund.
\newblock Improved subquadratic {3SUM}.
\newblock {\em Algorithmica}, 77(2):440--458, 2017.

\bibitem{GIMV03}
Martin Gavrilov, Piotr Indyk, Rajeev Motwani, and Suresh Venkatasubramanian.
\newblock Combinatorial and experimental methods for approximate point pattern
  matching.
\newblock {\em Algorithmica}, 38(1):59--90, 2004.

\bibitem{GS17}
Omer Gold and Micha Sharir.
\newblock Improved bounds for {3SUM}, k-{SUM}, and linear degeneracy.
\newblock In {\em 25th Annual European Symposium on Algorithms, {ESA} 2017,
  September 4-6, 2017, Vienna, Austria}, pages 42:1--42:13, 2017.

\bibitem{HB17}
Jacob~E. Goodman, Joseph O'Rourke, and Csaba~D. T\'oth, editors.
\newblock {\em Handbook of Discrete and Computational Geometry, Third Edition}.
\newblock {CRC} Press {LLC}, 2017.

\bibitem{GMO99}
Michael~T. Goodrich, Joseph S.~B. Mitchell, and Mark~W. Orletsky.
\newblock Approximate geometric pattern matching under rigid motions.
\newblock {\em {IEEE} Trans. Pattern Anal. Mach. Intell.}, 21(4):371--379,
  1999.

\bibitem{GP18}
Allan Gr{\o}nlund and Seth Pettie.
\newblock Threesomes, degenerates, and love triangles.
\newblock {\em J. {ACM}}, 65(4):22:1--22:25, 2018.

\bibitem{KLM19}
Daniel~M. Kane, Shachar Lovett, and Shay Moran.
\newblock Near-optimal linear decision trees for k-{SUM} and related problems.
\newblock {\em J. {ACM}}, 66(3):16:1--16:18, 2019.

\bibitem{LR97}
Mikl\'{o}s Laczkovich and Imre~Z. Ruzsa.
\newblock The number of homothetic subsets.
\newblock In Ronald~L. Graham and Jaroslav Ne{\v{s}}et{\v{r}}il, editors, {\em
  The Mathematics of Paul Erd\H{o}s II}, pages 294--302. Springer Berlin
  Heidelberg, 1997.

\bibitem{W18}
Virginia~Vassilevska Williams.
\newblock On some fine-grained questions in algorithms and complexity.
\newblock In {\em Proceedings of the ICM}, 2018.

\end{thebibliography}

\end{document}